\documentclass[letterpaper, 10 pt, conference]{ieeeconf}  %

\IEEEoverridecommandlockouts                              %
\usepackage{graphics} %
\usepackage{epsfig} %
\usepackage{amsmath} %
\usepackage{amssymb}  %
\usepackage{amsfonts}
\usepackage{mathtools}
\usepackage{xcolor}
\usepackage[acronym]{glossaries}
\let\labelindent\relax
\usepackage{enumitem}   
\makeatletter
\newcommand{\pushright}[1]{\ifmeasuring@#1\else\omit\hfill$\displaystyle#1$\fi\ignorespaces}
\makeatother

\makeatletter
\let\NAT@parse\undefined
\makeatother
\usepackage[hidelinks]{hyperref}
\usepackage[capitalise,nameinlink,sort]{cleveref}

\DeclareMathOperator*\uplim{\overline{lim}}

\DeclareMathOperator*{\argmin}{arg\,min}

\newacronym{ediss}{\emph{E-$\delta$-ISS}}{exponentially incrementally input-to-state stable}

\newtheorem{theorem}{Theorem}
\newtheorem{lemma}{Lemma}
\newtheorem{assumption}{Assumption}
\newtheorem{definition}{Definition}

\newcommand{\ak}[1]{{\color{black}  {#1}}}

\title{\LARGE \bf
On the Regret of Recursive Methods for\\
Discrete-Time Adaptive Control with Matched Uncertainty
}

\author{Aren Karapetyan, Efe C. Balta, Anastasios Tsiamis, Andrea Iannelli, and John Lygeros%
\thanks{This work has been supported by the Swiss National Science Foundation under NCCR Automation (grant agreement $51\text{NF}40\_180545$), the  European Research Council under the ERC Advanced grant agreement  $787845$ (OCAL) and by the German Research Foundation (DFG) under Germany’s Excellence
Strategy - EXC 2075 – 390740016.
}%
\thanks{A. Karapetyan, A. Tsiamis and J. Lygeros are with the Automatic Control Laboratory, Swiss Federal Institute of Technology (ETH Z\"urich), 8092 Z\"urich, Switzerland (E-mails: {\tt\footnotesize {akarapetyan, atsiamis, jlygeros}@ethz.ch).}}%
\thanks{E. C. Balta is with the Control and Automation Group, inspire AG, 8005 Zürich, Switzerland, and with the Automatic Control Laboratory {\tt\footnotesize efe.balta@inspire.ch}.}
\thanks{A. Iannelli is with the Institute for Systems Theory and Automatic Control, University of Stuttgart, Stuttgart 70569, Germany (E-mail: {\tt\footnotesize andrea.iannelli@ist.uni-stuttgart.de}).}%
}

\begin{document}

\maketitle
\thispagestyle{empty}
\pagestyle{empty}

\begin{abstract}

Continuous-time adaptive controllers for systems with a matched uncertainty often comprise  an online parameter estimator and a corresponding parameterized controller to cancel the uncertainty. However, such methods are often \ak{impossible to implement directly}, as they depend on an unobserved estimation error. \ak{We consider the equivalent discrete-time setting with a causal information structure, and propose a novel, online proximal point method-based adaptive controller, that under a sufficient excitation (SE) condition is asymptotically stable and achieves finite regret, scaling only with the time required to fulfill the SE}. We show the same also for the widely-used recursive least squares with exponential forgetting controller under a \ak{stronger persistence of excitation condition}.

\end{abstract}

\section{Introduction}

Adaptive control studies controllers that can adapt to or learn unmodelled changes in the dynamics \cite{aastrom2013adaptive}\ak{, while often guaranteeing perfect asymptotic tracking and/or parameter convergence  \cite{astolfi2008nonlinear, narendra2012stable}}. \ak{ The theory of online optimization was established concurrently with the developments in adaptive control, with the aim to} minimize an a priori unknown, sequentially revealed cost \cite{cesa2006prediction}. The connections between the two have recently been highlighted \cite{raginsky2010online, annaswamy2021historical}, showing that adaptive controllers are closely linked to an online cost optimization problem and are often proportional to an online estimation cost gradient. In this work, we consider \ak{such} controllers that aim to minimize a cost function while controlling the system\ak{, similar to the goal-oriented control formulation in \cite{fradkov1999nonlinear}}. 
\ak{Since} asymptotic stability alone does not automatically \ak{imply} performance guarantees on the cost \cite{nonhoff2023relation, karapetyan2023implications}\ak{, we also} provide finite-time guarantees by characterizing the \emph{regret} \cite{cesa2006prediction} of \ak{the} controller: a concept borrowed from online learning to quantify the additional cost incurred by the controller over a horizon due to partial model knowledge.

We consider systems with matched uncertainty \cite{astolfi2008nonlinear}, where the unknown dynamics is parameterized by an unknown parameter vector and a known, state-dependant feature matrix. In this setting, the uncertainty can be exactly canceled by the control input given knowledge of the true parameter. 
We aim to find causal adaptive controllers that estimate this parameter recursively while stabilizing the system and minimizing regret. 
Most continuous-time adaptive controllers in this setting utilize Lyapunov theory and provide only asymptotic stability guarantees, see for example \cite{aastrom2013adaptive, astolfi2008nonlinear, narendra2012stable}. Using the same tools, recently, a finite $\mathcal{O}(1)$ regret bound  \ak{was shown in \cite{gaudio2019connections} for linear systems, and  in \cite{boffi2021regret} for nonlinear time-varying systems}. 
In \cite{boffi2021regret}, the authors also extend the results to a discrete-time setting, achieving a $o(T)$,  sublinear-in-time, bound with a causal formulation, and a $\mathcal{O}(1)$ with a non-causal one. In contrast to our proposed method, the Lyapunov-based adaptive controllers in \cite{gaudio2019connections, boffi2021regret} are not realizable in discrete time, as the parameter update depends on the current, unobserved error, making it non-causal. The online gradient descent-based algorithm in \cite{boffi2021regret}, although causal, achieves ${o}(T)$ regret even in the deterministic case. The difficulty is inherent in the discrete-time formulation of the problem, which introduces a one-step delay in the error observation. In discrete time, the convenient Lyapunov function cancellation \cite{astolfi2008nonlinear} no longer holds and discrete-time Lyapunov functions are
known to be harder to construct \cite{akhtar2004logarithmic}. 

The recursive learning algorithm for the adaptive controller can be set up to relate to an online or ``running" cost. This setting is \ak{also} studied in the time-varying optimization literature providing asymptotic tracking and convergence guarantees for online costs, e.g., \cite{simonetto2016convergence, dall2019convergence}. However, in this line of work, the problem setup considers no underlying dynamics for the state, and the regret \ak{is generally} not studied.

We design causal, uncertainty-matching, recursive methods for discrete-time, nonlinear, time-varying systems.  Our contributions are threefold. \ak{Firstly}, we introduce a recursive proximal learning (RPL) parameter estimation algorithm based on the proximal point method \cite{parikh2014proximal}. Under a weak notion of persistence of excitation (PE)\ak{, which we refer to as sufficient excitation (SE) as in the continuous-time adaptive control litarature \cite{pan2023comparative},} RPL produces estimates of the unknown parameters that are contractive with respect to the true one. \ak{Secondly}, we show that RPL-based adaptive control achieves finite regret scaling with the time required to achieve the \ak{SE} condition. By contrast, in \cite{boffi2021regret} finite regret is only achieved through access to a Lyapunov function or by a non-causal controller.
Finally, we analyze the regret of the recursive least-squares with forgetting factor (RLSFF) \cite{johnstone1982exponential}. We show that while both RPL and RLSFF achieve finite regret with similar bounds\ak{, RLSFF requires a stronger PE condition.} We demonstrate the algorithms on a discrete-time model reference adaptive control (MRAC) \cite{dogan2020improving} example.

\textit{Notation}: The sets of positive real numbers, positive integers, and non-negative integers are denoted by $\mathbb{R}_{+}$, $\mathbb{N}_{+}$ and $\mathbb{N}$, respectively. For a given vector $x$, its  Euclidean norm is denoted by $\|x\|$.  The spectral norm of a square matrix $W$ is denoted by $\|W\|$, and the largest and smallest singular values and eigenvalues by $\sigma_{min}(W)$, $\sigma_{max}(W)$, and $\lambda_{min}(W)$, $\lambda_{max}(W)$ respectively. For a positive scalar $r\in \mathbb{R}_+$, we define the closed ball around the origin as $\mathcal{B}_r:=\{x\in \mathbb{R}^n \;| \;\|x\|\leq r\}$.

\section{Problem Statement}
\label{sec:problem_statement}

We consider nonlinear dynamical systems with matched uncertainty of the form
\begin{equation}
    \label{eq:nonlinear_matched}
    x_{k+1} = f_k(x_k) + B_k(x_k)\left(u_k - \ak{\alpha_k}(x_k)\right),
\end{equation}
where $f: \mathbb{N} \times \mathbb{R}^n  \rightarrow \mathbb{R}^n$ models the known nominal dynamics and satisfies $f_k({0}) = {0}$ for all $k\in \mathbb{N}$, $B_k: \mathbb{N}\times\mathbb{R}^n \rightarrow \mathbb{R}^{n\times m}$ is a known input matrix,  and\ak{ $\alpha_k:  \mathbb{N}\times\mathbb{R}^n \rightarrow \mathbb{R}^m$} is an \emph{unknown} mapping to the input space, capturing the uncertain dynamics of the model. We assume that  $\alpha_k$ can be exactly parameterized by possibly nonlinear, but known and bounded basis matrices, \ak{$\phi_k:  \mathbb{N}\times\mathbb{R}^n \rightarrow \mathbb{R}^{p\times m}$}, that map the state to a feature in $\mathbb{R}^{p\times m}$. The uncertainty is then given by $\alpha_k(x_k) = \ak{\phi_k^\top}(x_k)\theta^\star$, where $\theta^\star \in \mathbb{R}^{p}$ is an unknown, \ak{constant} parameter vector\footnote{We drop the explicit state dependence on state for $B_k:=B_k(x_k)$ and $\phi_k:= \ak{\phi_k}(x_k)$, wherever required for readability.}. We also assume $\|B_k(x_k)\ak{\phi_k^\top}(x_k)\| \leq b$ for all $k\in \mathbb{N}$, $x \in \mathbb{R}^n$ and some  $b \in \mathbb{R}_+$. The tracking problem under partial model knowledge can also be cast into this framework in the context of MRAC, as shown in \cref{sec:numerical}.

We characterize the performance of a given controller for system \eqref{eq:nonlinear_matched} over a time horizon $T$ by comparing it to a benchmark control signal $u^\star=\begin{bmatrix} 
    {u^\star_0}^\top & \hdots & {u^\star_{T-1}}^\top
\end{bmatrix}^\top$ that cancels the uncertainty perfectly and is defined below. We use regret to quantify the controller performance
\begin{equation}
    \label{eq:regret}
    \mathcal{R}_T({u}) = \sum_{k=0}^{T-1} c(x_k)- c(x^\star_k),
\end{equation}
where $u=\begin{bmatrix} 
    u_0^\top & \hdots & u_{T-1}^\top
\end{bmatrix}^\top$ denotes the control signal of a given controller, $x_k$ and $x^\star_k$ are state sequences due to the control signals $u$ and $u^\star$, respectively, and $c:\mathbb{R}^n \rightarrow \mathbb{R}$ is a stage cost.  We restrict our analysis to locally Lipschitz continuous costs and assume for any $R\in \mathbb{R}_+$, there exists  a $L_c \in \mathbb{R}_+$, such that for all $x,y \in \mathcal{B}_R$, 
    \begin{equation}
    \label{eq:lipschitz_cost}
        \|c(x) - c(y)\| \leq L_c \|x-y\|.
    \end{equation}

\begin{assumption}
\label{assum:system}
(System Dynamics)
\label{assum:ediss}
    The  nominal dynamics $x_{k+1} = f_k(x_k)$  are uniformly, globally \acrshort{ediss} \cite{angeli2002lyapunov}. In other words, there exist $c_0,c_w, r_w \in \mathbb{R}_+$ and $\rho \in (0,1)$, such that for any $x_0,y_0 \in\mathbb{R}^n$, $w_k \in \mathcal{B}_{r_w}$, and all $k \in \mathbb{N}$ the perturbed dynamics $y_{k+1} = f_k(y_k)+w_k$ satisfy
    \begin{equation*}
        \|x_k-y_k\| \leq c_0 \rho^k\|x_0-y_0\|+c_w\sum_{i=0}^{k-1}\rho^{k-i-1}\|w_i\|.
    \end{equation*}
 
\end{assumption}

The \acrshort{ediss} property characterizes only the known and autonomous part of the dynamics. It is known that if $f_k$ is smooth in the state and is contractive, then it is also \acrshort{ediss} \cite{tran2018convergence,boffi2021regret}. If $f_k$ is not smooth, then under milder conditions, we show in \cite{karapetyan2023closed} that exponential stability implies \acrshort{ediss}. %

\subsection{Online Estimation} We consider adaptive control laws that produce control inputs of the form $u_k =\ak{\phi_k^\top}(x_k)\theta_k$
where $\theta_k \in \mathbb{R}^{p}$ is an online estimate of $\theta^\star$ at time $k$. The closed-loop dynamics for a controller of this form is given by
\begin{equation}
    \label{eq:suboptimal}
    x_{k+1} = f_k(x_k) + \underbrace{B_k(x_k)\ak{\phi_k^\top}(x_k)\left(\theta_k-\theta^\star\right)}_{w_k(x_k)},
\end{equation}
where $w_k(x_k)$ can be thought of as a state-dependant artificial disturbance introduced due to an inexact uncertainty matching. One can also consider a benchmark counterfactual control input $u_k^\star: = \ak{\phi_k^\top}(x_k^{\star}) \theta^\star$ and the corresponding signal  $u^\star$ that perfectly cancels out the uncertainty in \eqref{eq:nonlinear_matched} leading to a benchmark state evolution 
\begin{equation}
\label{eq:optimal}
    x_{k+1}^\star = f_k(x_k^\star).
\end{equation}
Note that $u^\star$ is a counterfactual policy that is not realizable, as $\theta^\star$ is unknown. To simplify the analysis, we take $x_0=x_0^\star$. 

A standard setup in adaptive control design is to produce estimates $\theta_k$ that minimize the magnitude of $w_k(x_k)$, effectively steering the state of \eqref{eq:suboptimal} to that of \eqref{eq:optimal} thanks to the \acrshort{ediss} property. This is often done by setting up a least squares estimation problem, minimizing the following online estimation cost for all $ k \in \mathbb{N}_+$
\begin{align}
\label{eq:nominal_cost}
    \theta_k &= \argmin_\theta h_{k-1}(\theta), \\
    h_{k-1}(\theta): \!&= \!  \frac{1}{2}\sum_{i=0}^{k-1}\|B_i\phi_i^\top\left(\theta-\theta^\star\right)\|^2 = \frac{1}{2}\|\Phi_{k}{\theta} -Y_{k}\|^2,\notag \\
    \nonumber
        Y_k &= \begin{bmatrix}
        y_0 \\
        y_1 \\
        \vdots \\
        y_{k-1}
    \end{bmatrix}, \qquad
    \Phi_k = \begin{bmatrix}
        B_0\phi_0^\top\\
        B_1\phi_1^\top\\
        \vdots \\
        B_{k-1}\phi^\top_{k-1}
    \end{bmatrix},\\
    \nonumber
    y_i &= \ak{-x_{i+1} + f_i(x_i) + B_i\phi_i^\top \theta_i = B_i\phi_i^\top \theta^\star, \quad \forall i \in \mathbb{N}}.
\end{align}
Note that at time $k$, the latest available online estimation cost is $h_{k-1}$, as $y_k$ is not yet available.
Solving \eqref{eq:nominal_cost} directly requires computing an expensive matrix inverse at each time step and maintaining an increasing memory. In this work, we study two recursive methods that alleviate this issue, both optimizing an online cost related to \eqref{eq:nominal_cost}. 
In particular, we introduce and analyze a recursive proximal learning method in Section \ref{sec:rpl}, and also study the well-established recursive least squares with forgetting factor \cite{johnstone1982exponential} in Section \ref{sec:rls}.

\subsection{Regret}
Speed or velocity-gradient controllers are introduced in \cite{fradkov1999nonlinear} to solve the continuous-time equivalent of the online estimation problem \eqref{eq:nominal_cost} in the framework of (integral) goal-oriented control. It is shown that these methods achieve finite estimation error in the limit, under the assumption that the controller at time $t$ can depend on the estimation error ${\theta}_t-\theta^\star$. Recent works \cite{gaudio2019connections, boffi2021regret} noted that if an exponential Lyapunov function exists for the nominal system, a speed-gradient descent-based uncertainty matching controller achieves finite quadratic cost on the state. In particular, the suboptimal closed-loop state trajectory $x_t$ satisfies
 \begin{equation*}
     J_\mathrm{cont}(x_0,u):=\lim_{T\rightarrow \infty}\int_{0}^{T}\|x_t\|^2dt = \mathcal{O}(1).
 \end{equation*}
Moreover, the existence of the exponential Lyapunov function also implies that $J_{\mathrm{cont}}(x_0,u^\star)$ for the optimal state evolution is finite. Defining continuous-time regret as $
    \mathcal{R}_{\mathrm{cont}}(u):= J_{\mathrm{cont}}(x_0,u) - J_{\mathrm{cont}}(x_0,u^\star)
$, an $\mathcal{O}(1)$ bound \ak{for} the extra cost accumulated by the adaptive control law as compared to the benchmark follows directly. In \cite{boffi2021regret} a discrete-time version of the speed-gradient algorithm is introduced with a  finite cost bound, but with the same assumption that at time $k$ the error $\theta_k-\theta^\star$ is available, making the controller non-causal. Without this assumption, the authors achieve $o(T)$ sublinear regret using an online learning toolkit. 
In this work, we consider causal controllers that minimize \eqref{eq:nominal_cost} at time $k$ with access only to $y_{k-1}$. To characterize the performance of such an adaptive, uncertainty-matching controller, we analyze the regret \eqref{eq:regret}.

\section{Recursive Proximal Learning}
\label{sec:rpl}

The recursive proximal learning algorithm   corresponds to the $\varepsilon$-scaled proximal operator of \eqref{eq:nominal_cost} evaluated at the \emph{previous} estimate
\begin{equation}
\label{eq:proximal_update}
\begin{split}
    \theta_{k} &= \mathrm{prox}_{\varepsilon h_{k-1}}\left(\theta_{k-1}\right)\\
    &=\argmin_{\theta}\left(h_{k-1}(\theta) + \frac{\varepsilon}{2}\|\theta-\theta_{k-1}\|^2\right),
    \end{split}
\end{equation}
for some $\theta_0 \in \mathbb{R}^p$ and $\varepsilon \in \mathbb{R}_+$. From the online learning perspective \cite{cesa2006prediction}, $\theta_k$ can be thought of as the minimizer of the latest available cost $g_{k-1}$ at time $k$\ak{, where we define}
\begin{equation}
    \label{eq:proximal_cost}
    g_{k-1}(\theta) \ak{:=} h_{k-1}(\theta) + \frac{\varepsilon}{2}\|\theta-\theta_{k-1}\|^2.
\end{equation}
While a closed-form solution to \eqref{eq:proximal_update} exists, it requires memory increasing with time. The following equivalent set of updates to \eqref{eq:proximal_update} \ak{provides a recursive implementation}. In particular, for all $k\in \mathbb{N}$, the following is equivalent to \eqref{eq:proximal_update} %
\begin{equation}
\label{eq:rpl}
\begin{split}
P_{k+1}^{-1} &= P_k^{-1} + \phi_kB_k^\top B_k \phi_k^\top, \\
H_{k+1} &= H_{k} + \phi_k B_k^\top B_k\phi_k^\top\\
s_{k+1} &= s_{k} + \phi_kB_k^\top y_k\\
\theta_{k+1} &= \theta_k - P_{k+1}\left(H_{k+1} \theta_k - s_{k+1}\right),
\end{split}
\end{equation}
initializing with $H_0=\boldsymbol{0}_{p \times p}$, $s_0=\boldsymbol{0}_{p\times 1}$ and $P_0^{-1} = \varepsilon I_p$. %
\subsection{Contraction in Parameter Space}
We define \ak{SE} as follows.

\begin{definition}
\label{def:pe}
    A matrix signal $\phi_k:\mathbb{N} \rightarrow \mathbb{R}^{p \times n}$ is called \ak{sufficiently} exciting if \ak{there exist $T_s \in \mathbb{N}$, and $\delta \in \mathbb{R}_+$}  such that
    \begin{align*}
      0 <\delta I_p \leq \sum_{i=0}^{T_s}\phi_i\phi_i^\top.
    \end{align*}
\end{definition}

We show below that under a \ak{SE} assumption, the RPL update \eqref{eq:proximal_update} enjoys a contraction property on the parameter; moreover, under an additional \ak{boundedness} assumption, the lifted input vector $\Phi_k \theta_k$ is also contractive.
\begin{assumption}
    \label{assum:beta_bound}
    There exists a $\beta \in \mathbb{R}_+$ such that
$\lim_{T\rightarrow \infty}\sum_{i=0}^{T}\phi_i B_i^\top B_i\phi_i^\top\leq \beta I_p
  $
\end{assumption}

\begin{lemma}
\label{lem:rpl}
    Assume $B_k \phi_k^\top$ is \ak{sufficiently} exciting, then  the RPL estimate  \eqref{eq:proximal_update} satisfies
    \begin{equation}
        \label{eq:theta_contraction}
    \begin{split}
        \|{\theta}_{k} - \theta^\star\| &\leq \eta\|{\theta}_{k-1} - \theta^\star\|, \qquad \forall k\geq T_s\\
        \|{\theta}_{k} - \theta^\star\| &\leq \|{\theta}_{k-1} - \theta^\star\|, \qquad \forall 0< k < T_s
    \end{split}
    \end{equation}
   where $\eta = \frac{\varepsilon}{\delta+\varepsilon} \in (0,1)$. Moreover, if Assumption \ref{assum:beta_bound} holds, and $\varepsilon < \frac{\delta\sqrt{\delta}}{\sqrt{\beta}-\sqrt{\delta}}$ then there exists a $\gamma \in (0,1)$ such that %
   \begin{align*}
       \|\Phi_{k+1} \left({\theta}_{k}-\theta^\star\right)\| &\leq \gamma\|\Phi_{k} \left({\theta}_{k-1}-\theta^\star\right)\|, \qquad \forall k \geq T_s,\\
       \|\Phi_{k+1} \left({\theta}_{k}-\theta^\star\right)\| &\leq \sqrt{\beta} \|\theta_0-\theta^\star\|, \qquad \forall 0\leq k < T_s.
   \end{align*}
 \end{lemma}
 \begin{proof}
 The proximal update \eqref{eq:proximal_update} can equivalently be represented by a single online Newton step update \cite{hazan2006logarithmic} on \eqref{eq:proximal_cost}. At time $k$, it is equivalently represented by
 \begin{equation*}
    \theta_{k} = \theta_{k-1} - \left(\nabla^2g_{k-1}(\theta)\Bigr|_{\theta=\theta_{k-1}}\right)^{-1}\nabla g_{k-1}(\theta)\Bigr|_{\theta=\theta_{k-1}}.
 \end{equation*}
 Subtracting the true $\theta^\star$ from both sides, denoting $\Tilde{\theta}_k:=\theta_k-\theta^\star$, noting that $Y_k = \Phi_k \theta^\star$ and recalling the definition of $h_{k-1}$ from \eqref{eq:nominal_cost}
    \begin{align*}
        \tilde{\theta}_{k} &= \tilde{\theta}_{k-1}-\left(\Phi_{k}^\top \Phi_{k} + \varepsilon I_p\right)^{-1}\Phi_{k}^\top \left(\Phi_{k}\theta_{k-1}-Y_{k}\right)\\
        &= \tilde{\theta}_{k-1}-\left(\Phi_{k}^\top \Phi_{k} + \varepsilon I_p\right)^{-1}\Phi_{k}^\top\Phi_{k}\tilde{\theta}_{k-1}\\
        &= \left(I_p - \left(\Phi_{k}^\top \Phi_{k} + \varepsilon I_p\right)^{-1}\Phi_{k}^\top\Phi_{k} \right)\tilde{\theta}_{k-1}\\
        &:= \left(I_p - M_{k}\right)\Tilde{\theta}_{k-1},%
    \end{align*}
where $M_k = \left(\Phi_{k}^\top \Phi_{k} + \varepsilon I_p\right)^{-1}\Phi_{k}^\top\Phi_{k} $. Using the orthonormality property of singular value decomposition of symmetric matrices and the \ak{sufficient} excitation of $B_k\phi_k$, for all $k\geq T_s$ 
    \begin{align*}
        \sigma_{min}\left(M_{k}\right) &\geq \frac{\delta}{\varepsilon + \delta} \in (0,1),\\
        \eta_k:= \|I_p - M_{k}\|&\leq \frac{\varepsilon}{\delta+\varepsilon} \in (0,1),
    \end{align*}
    obtaining the contraction result for parameter since $\sigma_{max}(M_k) \in(0,1)$. For all $0< k < T_s$, the \ak{SE} condition is not fulfilled, hence $\sigma_{min}(M_k) = 0$ and $\|\Tilde{\theta}_k\|\leq \| \Tilde{\theta}_{k-1}\|$. 
    
    To prove the second part, note that for a full column rank matrix $A$, 
        $\|A\Tilde{\theta}_{k+1}\| = \|\Tilde{\theta}_{k+1}\|_{A^\top A}$, 
    and ${{\lambda_{min}(A^\top A)}} = {\sigma_{min}(A^\top A)}$ , and likewise for $\lambda_{max}$. Then, for all $k\geq T_s$
    \begin{align*}
    &\sigma_{min}\left(\left(\Phi_{k+1}^\top\Phi_{k+1}\right)^{-1}\right)\|\Tilde{\theta}_{k}\|^2_{\Phi_{k+1}^\top\Phi_{k+1}} \leq  \|\Tilde{\theta}_{k}\|^2 \\
    &\leq \eta^2 \|\Tilde{\theta}_{k-1}\|^2 \leq \eta^2 \sigma_{max}\left(\left(\Phi_{k}^\top\Phi_{k}\right)^{-1}\right) \|\Tilde{\theta}_{k-1}\|_{\Phi_{k}^\top\Phi_{k}}^2.
    \end{align*}
    It follows that
    \begin{align}
         \|\Phi_{k+1}\Tilde{\theta}_{k}\| \leq \eta \sqrt{\frac{\sigma_{max}\left(\Phi_{k+1}^\top\Phi_{k+1}\right)}{\sigma_{min}\left(\Phi_{k}^\top\Phi_{k}\right)}} \|\Phi_{k} \Tilde{\theta}_{k-1}\|,
    \end{align}
    where we used the fact that $\sigma_{max}\left((A^\top A)^{-1}\right) = 1/\sigma_{min}(A^\top A)$ and $\sigma_{min}\left((A^\top A)^{-1}\right) = 1/\sigma_{max}(A^\top A)$. 
    The proof is completed by noting that for all $k\geq T_s$
    \begin{equation*}
        \gamma:= \eta \sqrt{\frac{\sigma_{max}\left(\Phi_{k+1}^\top\Phi_{k+1}\right)}{\sigma_{min}\left(\Phi_{k}^\top\Phi_{k}\right)}}
        \leq \frac{\varepsilon\sqrt{\beta}}{\varepsilon\sqrt{\delta} +\delta\sqrt{\delta}}<1,
    \end{equation*}
    given $\varepsilon < \frac{\delta\sqrt{\delta}}{\sqrt{\beta}-\sqrt{\delta}}$. Finally, for all $k \in \mathbb{N}$, $\|\Phi_{k+1}\Tilde{\theta}_k\|\leq\|\Phi_{k+1}\|\|\Tilde{\theta}_k\|\leq \sqrt{\beta}\|\Tilde{\theta}_0\|$.
 \end{proof}

\subsection{Closed-loop Analysis}

Consider the suboptimal state evolution \eqref{eq:suboptimal} under the RPL controller $u_k^{\mathrm{RPL}} = \phi_k^\top(x_k)\theta_k$, with $\theta_k$ given in \eqref{eq:proximal_update}. We can represent it equivalently as 
\begin{equation}
    \label{eq:suboptimal_rpl}
    x_{k+1} = f_k(x_k)+ \underbrace{S_k\left(\Phi_{k+1}{\theta}_k-\Phi_{k+1}\theta^\star\right)}_{w_k(x_k)},
\end{equation}
where  $S_k:= 
    \begin{bmatrix}
        \boldsymbol{0} & \hdots & \boldsymbol{0} & I_{n},
    \end{bmatrix} \in \mathbb{R}^{n \times n(k+1)}$
is a selector matrix that selects the last component of $\Phi_{k+1}\left({\theta}_k-\theta^\star\right)$, corresponding to $B_k\phi_k^\top(\theta_k-\theta^\star)$. Similarly, the optimal state evolution under $u^\star$ is given by
\begin{equation}
    \label{eq:optimal_rpl}
    x_{k+1}^\star =f_k(x^\star_k) +\underbrace{S_k\left(\Phi^\star_{k+1}\theta^\star-\Phi^\star_{k+1}\theta^\star\right)}_{\boldsymbol{0}},
\end{equation}
where $\Phi^\star_{k+1}:= \Phi(x^\star_{[0,\hdots,k]})$. In the following theorem, we characterize the closed-loop system \eqref{eq:suboptimal_rpl} in terms of asymptotic stability and regret. In particular, we show that the regret of the RPL controller scales with  $\mathcal{O}\left(\|\Tilde{\theta}_0\| T_s\right)$.

\begin{theorem}
\label{thm:rpl}
    \ak{ Let Assumption \ref{assum:ediss} hold, and $B_k\phi_k^\top(x_k)$ be sufficiently exciting under the RPL adaptive controller, then this controller} 
    \begin{enumerate}
        \item[(i)] renders the closed-loop system asymptotically stable,
        \item[(ii)] achieves finite regret
    \end{enumerate}
        \begin{equation*}
\mathcal{R}_T(\ak{u^{\mathrm{RPL}}})\leq c_w bL_c\|\Tilde{\theta}_0\|\left(\frac{T_s}{1-\rho}\!+\!\frac{\rho^T\!+\!(1-\eta)\rho+\eta}{(1-\rho)^2(1-\eta)}\right),
    \end{equation*} 
for the $\eta$ established in Lemma \ref{lem:rpl}. If in addition, $B_k\phi_k$ also satisfy  Assumption \ref{assum:beta_bound} and $\varepsilon < \frac{\delta \sqrt{\delta}}{\sqrt{\beta} - \sqrt{\delta}}$, then
        \begin{equation*}
\mathcal{R}_T(\ak{u^{\mathrm{RPL}}})\leq c_w c_{p}L_c\|\Tilde{\theta}_0\|\left(\frac{T_s}{1-\rho}\!+\!\frac{\rho^T\!+\!(1-\gamma)\rho+\gamma}{(1-\rho)^2(1-\gamma)}\right),
    \end{equation*} 
 where $c_p := \|\Phi_{T_s}\| \leq \sqrt{\beta}$ and $\gamma$ is defined in Lemma \ref{lem:rpl}.
\end{theorem}

\begin{proof}
  To show (i), recall the definition of \acrshort{ediss}, and the fact that $f_k({0}) = {0}$ for all $k\in \mathbb{N}$, then consider the suboptimal state evolution of the nonlinear system under the RPL adaptive controller \eqref{eq:suboptimal_rpl} for all $k\in \mathbb{N}_+$, recalling that $\Tilde{\theta}_k = \theta_k - \theta^\star$
  \begin{equation}
  \label{eq:x_bounded}
    \begin{split}
              \|x_k\| &\leq c_0\rho^k\|x_0\| + c_w\sum_{i=0}^{k-1}\rho^{k-i-1}\|B_i \phi_i\Tilde{\theta}_i\|\\
              &< c_0\|x_0\|+\frac{c_w b\| \Tilde{\theta}_0\|}{1-\rho} ,
    \end{split}
    \end{equation}
    where we use the submultiplicativity property of the norms, the bound on $B_k\phi_k$ and non-expansivity of $\|\Tilde{\theta}_k\|$ from Lemma \ref{lem:rpl}. \ak{The asymptotic stability follows from the input-to-state stability \cite{jiang2001input} of the closed-loop state in \eqref{eq:x_bounded}:  by Lemma \ref{lem:rpl}, $\lim_{k \rightarrow \infty}\|\Tilde{\theta}_k\| = 0$ implies that $\lim_{k \rightarrow \infty}\|x_k\| = 0$}.
   
   Next, we show (ii), by first noting that since $x_k$ and $x_k^\star$ are bounded from \eqref{eq:x_bounded}, it follows from \eqref{eq:lipschitz_cost}
\begin{align*}
    &\mathcal{R}_T(\ak{u^{\mathrm{RPL}}})\leq \underbrace{L_c\sum_{k=T_s+1}^{T-1}\|x_k-x_k^\star\|}_{\text{I}} + \underbrace{L_c\sum_{k=0}^{T_s}\|x_k-x_k^\star\|}_{\text{ II}},
\end{align*}
where II represents the cost in the initial $T_s$-long non-PE phase, and  I the cost improvement afterwards.

Using the \acrshort{ediss} condition on the nominal system  and noting that  $x_k$ and $x_k^*$, given by \eqref{eq:suboptimal_rpl}  and $\eqref{eq:optimal_rpl}$, respectively, are equal in the first time step, for all $k>T_s$
\begin{equation}
\label{eq:x_bound_I}
\begin{split}
    &\|x_k-x_k^\star\| \leq c_wb\sum_{i=0}^{k-1}\rho^{k-i-1}\|\Tilde{\theta}_i\|\\
    &\leq c_w b\sum_{i=1}^{T_s}\rho^{k-i}\|\Tilde{\theta}_{i-1}\| + c_wb\sum_{i=T_s}^{k-1}\rho^{k-i-1}\|\Tilde{\theta}_{i}\|\\
    &\leq c_wb\rho^{k-T_s}\|\Tilde{\theta}_{0}\|\sum_{i=1}^{T_s}\rho^{T_s-i} + c_wb\sum_{i=T_s}^{k-1}\rho^{k-i-1}\|\Tilde{\theta}_{i}\|\\
    &\leq \rho c_wc_{T_s}b\|\Tilde{\theta}_0\|\rho^{k-T_s} +  c_wb\|\Tilde{\theta}_0\|\sum_{i=T_s}^{k-1}\rho^{k-i-1}\eta^{i-T_s+1},
\end{split}
\end{equation}
where, \ak{defining $c_{T_s}= \frac{1-\rho^{T_s}}{1-\rho}$}, the second inequality follows from the boundedness of $B_k\phi_k$, the third from the non-expansivity of $\|\Tilde{\theta}_k\|$ from Lemma \ref{lem:rpl} and rearrangement of the constants,  and the last one from the contraction of $\|\Tilde{\theta}_k\|$. For $0<k\leq T_s$, we may only use the non-expansive result from Lemma \ref{lem:rpl} and hence
\begin{equation}
\label{eq:x_bound_II}
    \|x_k-x_k^\star\| \leq c_w b\sum_{i=1}^{k}\rho^{k-i}\|\Tilde{\theta}_{i-1}\|\leq c_wb\|\Tilde{\theta}_0\|\sum_{i=1}^{k}\rho^{k-i}.
\end{equation}
Using the Cauchy Product inequality defined for two finite series $\{a_i\}_{i=1}^T$ and $\{b_i\}_{i=1}^T$
\begin{equation*}
    \label{eq:cauchy_product}\textstyle{
    \sum_{i=0}^T\left|\sum_{j=0}^{i}a_jb_{i-j}\right| \leq  \left(\sum_{i=0}^T|a_i|\right) \left(\sum_{j=0}^T|b_j|\right)},
\end{equation*} and summing \eqref{eq:x_bound_II} and \eqref{eq:x_bound_I} over $T_s-1$ and $T-T_s-2$, respecitvely, we get
\begin{align*}
    \text{ II}&\leq c_wbL_c\|\Tilde{\theta}_0\|\sum_{k=0}^{T_s-1}\sum_{i=0}^k\rho^{k-i}\leq c_wc_{T_s}T_sbL_c\|\Tilde{\theta}_0\|\\
    \text{I}
    &\leq c_wbL_c\|\tilde{\theta}_0\|\left(\rho c_{T_s}\!\sum_{k=1}^{T-T_s-1}\rho^{k}\!+\!\sum_{k=0}^{T-T_s-2}\sum_{i=0}^{k}\rho^{k-i}\eta^{i+1}\right)\\
    &\leq \rho c_wbL_c\|\Tilde{\theta}_0\|\frac{\rho^T + \rho}{(1-\rho)^2}\\
    &\quad + c_wb\eta L_c\|\Tilde{\theta}_0\| \frac{\left(1-\rho^{T-T_s-1}\right)\left(1-\eta^{T-T_s-1}\right)}{\left(1-\rho\right)\left(1-\eta\right)}.
\end{align*}
The first regret bound is achieved by combining the terms. 

Now, if Assumption \ref{assum:beta_bound} is satisfied and $\varepsilon < \frac{\delta \sqrt{\delta}}{\sqrt{\beta} - \sqrt{\delta}}$, one can make use of the lifted input contraction in Lemma \ref{lem:rpl}. In particular, for all $k>T_s$, referring to \eqref{eq:suboptimal_rpl} and \eqref{eq:optimal_rpl} again
   \begin{equation}
       \label{eq:x_contractive}
   \begin{split}
       &\|x_k-x_k^\star\| \leq c_w\sum_{i=0}^{k-1}\rho^{k-i-1}\|w_i(x_i)\|\\
       &\leq c_w\|S_k\|\sum_{i=0}^{k-1}\rho^{k-i-1}\|\Phi_{i+1}\Tilde{\theta}_{i}\|\\
       &\leq c_w\sum_{i=1}^{T_s}\rho^{k-i}\|\Phi_{i}\Tilde{\theta}_{i-1}\| + c_w\sum_{i=T_s}^{k-1}\rho^{k-i-1}\|\Phi_{i+1}\Tilde{\theta}_{i}\|\\
       &\begin{split}
        \leq c_w\sum_{i=1}^{T_s}\rho^{k-i}&\|\Phi_{i}\Tilde{\theta}_{i-1}\| \\
        &+c_w\|\Phi_{T_s}\Tilde{\theta}_{T_s-1}\|\sum_{i=T_s}^{k-1}\rho^{k-i-1}\gamma^{i-T_s+1}
        \end{split}
        \\
       &\begin{split}
        =c_w\rho^{k-T_s}\sum_{i=1}^{T_s}&\rho^{T_s-i}\|\Phi_{i}\Tilde{\theta}_{i-1}\|\\ &+c_w\|\Phi_{T_s}\Tilde{\theta}_{T_s-1}\|\sum_{i=T_s}^{k-1}\rho^{k-i-1}\gamma^{i-T_s+1}
       \end{split}
       \\
       &\leq \rho c_wc_pc_{T_s}\|\Tilde{\theta}_0\|\rho^{k-T_s}+c_wc_p\|\Tilde{\theta}_0\| \sum_{i=T_s}^{k-1}\rho^{k-i-1}\gamma^{i-T_s+1},
   \end{split}
   \end{equation}
   where the second inequality follows from the submultiplicative property of the norms, the third from the fact that $\|S_k\|=1$, and the last two from Lemma \ref{lem:rpl} and by denoting $c_{T_s}:= \frac{1-\rho^{T_s}}{1-\rho}$. Similarly, for $0<k\leq T_s$
   \begin{equation}
       \label{eq:x_nonexpansive}
    \|x_k-x_k^\star\| \leq c_w\sum_{i=1}^{k}\rho^{k-i}\|\Phi_{i}\Tilde{\theta}_{i-1}\|\leq c_wc_p\|\Tilde{\theta}_0\|\sum_{i=1}^{k}\rho^{k-i}.
    \end{equation}
We then  obtain new bounds for $\text{I}$ and $\text{II}$
\begin{align*}
    \text{II}&\leq c_wc_p\|\Tilde{\theta}_0\|L_c\sum_{k=0}^{T_s-1}\sum_{i=0}^k\rho^{k-i}\leq T_s \|\Tilde{\theta}_0\|c_wc_{T_s}c_pL_c,
\end{align*}
where the first inequality follows from \eqref{eq:x_nonexpansive} and the second one from the Cauchy Product inequality.
Using \eqref{eq:x_contractive}
\begin{align*}
    &\text{I}\!\leq\! c_wc_p\|\Tilde{\theta}_0\|L_c \left(\rho c_{T_s}\sum_{k=1}^{T-T_s-1}\rho^k\!+\! \sum_{k=0}^{T-T_s-2}\!\sum_{i=0}^{k}\rho^{k-i}\gamma^{i+1}\right) \\
    &\leq  \rho c_wc_pc_{T_s}\|\Tilde{\theta}_0\|L_c\frac{\rho - \rho^{T-T_s}}{1-\rho}\\
    &\qquad  +c_wc_p\|\Tilde{\theta}_0\| \gamma L_c \frac{\left(1-\rho^{T-T_s-1}\right)\left(1-\gamma^{T-T_s-1}\right)}{\left(1-\rho\right)\left(1-\gamma\right)}.
\end{align*}
Noting that the first term in the above summation is upper bounded by $\rho c_wc_p\|\Tilde{\theta}_0\|L_c\left(\frac{\rho^T +\rho}{(1-\rho)^2}\right)$ and combining I and II results in the second stated bound.

\end{proof}

The theorem states that given the assumptions on the nominal system, both \eqref{eq:suboptimal_rpl} and \eqref{eq:optimal_rpl} state trajectories converge to the origin. It is important to note that no further restrictions are put on the stage costs $c$, other than local Lipschitz continuity. Hence, the finite regret result of the theorem does not imply the cost $\sum_{k=0}^{T-1}c(x_k)$ is minimized, but the relative cost performance is bounded. Moreover, it shows that this performance, in terms of regret scales with the initial exploratory time period $T_s$, and cannot be avoided, as during these transients the parameter updates are only nonexpansive by Lemma \ref{lem:rpl}. If the additional assumptions on $\varepsilon$ and $B_k\phi_k$ are satisfied, then the constant term is updated from a uniform bound $b$ to $c_p:=\|\Phi_{T_s}\|$, which relates to the regression matrix only in the initial period. 

In the following lemma, we show that Lipschitz continuity of the basis matrices and $\phi({0})=\boldsymbol{0}_{p \times m}$ is a sufficient condition to fulfill Assumption \ref{assum:beta_bound}.

\begin{lemma}
    Let $B_k \phi_k$ be \ak{sufficiently} exciting, \ak{  $\phi_k({0})=\boldsymbol{0}_{p \times m}$, and $\phi_k$ be uniformly $L-$Lipschitz continuous, that is, there exists a $L\in \mathbb{R}_+$ such that
    \begin{equation*}
        \|\phi_k(x)-\phi_k(y)\| \leq L\|x-y\|. \quad \forall x,y \in \mathbb{R}^n, \forall k \in \mathbb{N},
    \end{equation*}}
    then Assumption \ref{assum:beta_bound} holds.
\end{lemma}
\begin{proof}
    Given \ak{SE} holds,  summing the first bound in \eqref{eq:x_bounded} over a horizon $T$ and taking the limit we get
  \begin{align*}
      &\uplim_{T\rightarrow \infty}\sum_{k=0}^T\|x_k\|\leq \frac{c_0\|x_0\|}{1-\rho}+c_w b\sum_{k=1}^{T_s}\sum_{i=1}^{k-1}\rho^{k-i}\|\Tilde{\theta}_{i-1}\|\\
      &<\uplim_{T\rightarrow \infty}c_w b\|\Tilde{\theta}_0\| \left(\rho c_{T_s}\sum_{k=1}^{T-T_s-1}\rho^k +\sum_{k=0}^{T-T_s-2}\sum_{i=0}^{k}\rho^{k-i}\eta^{i+1}\right)\\
      &<  \frac{\rho^2 c_w c_{T_s}b\|\Tilde{\theta}_0\|}{1-\rho}+ \frac{c_wb\|\Tilde{\theta}_0\|\eta}{(1-\rho)(1-\eta)}, %
  \end{align*}
    where the first inequality follows from the \acrshort{ediss} definition and geometric series, the second from Lemma \ref{lem:rpl} and the last from the Cauchy Product inequality. The result then follows from the boundedness of $B_k$ and $\phi_k$ and \ak{$\|\phi_k(x)\|\leq L\|x\|$.}
\end{proof}

\section{RLS with Exponential Forgetting}
\label{sec:rls}
The RLSFF algorithm minimizes the online cost 
\begin{equation}
\begin{split}
    \label{eq:rls_ff}
    g_{k-1}^f(\theta) := \frac{1}{2}\sum_{i=0}^{k-1}\lambda^{2(k-1-i)}&\|B_i\phi_i^\top \left(\theta-\theta^\star\right)\|^2\\ 
    &+ \frac{\lambda^{2k} \varepsilon}{2}\|\theta - \theta_0\|^2,    
    \end{split}
\end{equation}
 for a fixed exponential forgetting factor $\lambda^2 \in (0,1)$, where the first term can be thought of as the discounted version of \eqref{eq:nominal_cost} and the second, as a discounted regularizer on the \emph{initial} parameter error. The update can be recursively implemented as follows \cite{bruggemann2021exponential} for all $k \in \mathbb{N}$
\begin{align}
\label{eq:rls_P_update}
        P_{k+1}^{-1} &= \lambda^2 P_{k}^{-1} + \phi_{k}B_{k}^\top B_{k}\phi_{k}^\top,\\
            \theta_{k+1} &= \theta_{k} - P_{k+1} \phi_{k} B_{k}^\top(B_k\phi_{k}^\top\theta_{k} - y_{k}),
            \label{eq:rls_theta_update}
    \end{align}
where $\theta_0 \in \mathbb{R}^p$ is an initial guess and $P_0^{-1} = \varepsilon I_p$.  The convergence of the above update requires a persistence of excitation condition in the following sense \cite{johnstone1982exponential}.
\begin{definition}
\label{def:strong_pe}
 A matrix signal $\phi_k:\mathbb{N} \rightarrow \mathbb{R}^{p \times n}$ is called  persistently exciting  if it is bounded and there exist $ \delta \in \mathbb{R}_+$,  and $T_s \in \mathbb{N}$ such that for all $k_0 \in \mathbb{N}$
    \begin{equation*}
     \delta I_p \leq \sum_{i=k_0}^{k_0+T_s}\phi_i\phi_i^\top.
    \end{equation*}
\end{definition}
The PE condition in Definition \ref{def:strong_pe}, commonly found in the adaptive control literature \cite{aastrom2013adaptive, johnstone1982exponential},  is a stronger version of \ak{SE} in Definition \ref{def:pe}. The latter requires the positive definiteness of the regressor matrix only in \ak{some finite interval of time}, rather than, \ak{persistently}, in all \ak{$T_s$-long} time intervals. 
\begin{lemma}
\label{lem:rls}
    Assume $B_k\phi_k^\top$ is  persistently exciting,  then the RLSFF update \eqref{eq:rls_theta_update} satisfies
    \begin{align*}
        \|\theta_k-\theta^\star\| &\leq c_r\lambda^{k-T_s}\|\theta_0-\theta^\star\|,  &&\forall k\geq T_s\\
         \|\theta_k-\theta^\star\| &\leq \|\theta_0-\theta^\star\|,  &&\forall 0<k<T_s
    \end{align*}
    where $c_r^2 = {\varepsilon\left(\lambda^{2T_s}-\lambda^{-2)}\right)}/\left(\delta\left(1-\lambda^{-2}\right)\right)$.
\end{lemma}
The non-expansive result follows from \eqref{eq:rls_theta_update} by noting that 
\begin{equation*}
\ak{
\Tilde{\theta}_{k+1} = \left(I_p - P_{k+1}\phi_kB_k^\top B_k \phi_k^\top\right)\Tilde{\theta}_{k}, \quad \forall k>0,}
\end{equation*}
and that $\sigma_{min}\left(P_{k+1}\phi_kB_k^\top B_k \phi_k^\top \right)\geq 0$ and $\sigma_{max}\left(P_{k+1}\phi_kB_k^\top B_k \phi_k^\top \right)<1$. The proof for the exponential stability for $k\geq T_s$ can be found in \cite{bruggemann2021exponential} and is omitted. 

Consider the RLSFF controller $u_k^{\mathrm{RLS}} = \phi_k^\top(x_k)\theta_k$ with $\theta_k$ given by the recursive law~\eqref{eq:rls_P_update}-\eqref{eq:rls_theta_update}.

\begin{theorem}
\label{thm:rls}
   \ak{ Let Assumption \ref{assum:system} hold, and  $B_k\phi_k^\top(x_k)$ be persistently exciting under the RLSFF adaptive controller, then this controller }
    \begin{enumerate}
        \item[(i)] renders the closed-loop system asymptotically stable
        \item[(ii)] achieves finite regret
    \end{enumerate}
    \begin{equation*}    \mathcal{R}_T(\ak{u^{\mathrm{RLS}}})\leq c_wbL_c\|\Tilde{\theta}_0\|\left(\frac{T_s}{1-\rho}\!+\!\frac{c_r\left(\rho^T+1\right)}{(1-\rho)^2(1-\lambda)} \right).
\end{equation*}
\end{theorem}
\begin{proof} 
Consider the benchmark state $x_k^\star$ evolution in \eqref{eq:optimal} under perfect uncertainty matching and the RLSFF closed-loop state given by \eqref{eq:suboptimal} with $\theta_k$ update in \eqref{eq:rls_theta_update}. Using the \acrshort{ediss} condition on the nominal system  and noting that  $x_k$ and $x_k^*$ are equal in the first time step, for all $k>T_s$
\begin{align*}
    &\|x_k-x_k^\star\| \leq c_wb\sum_{i=0}^{k-1}\rho^{k-i-1}\|\Tilde{\theta}_i\|\\
    &\leq \rho c_wc_{T_s}b\|\Tilde{\theta}_0\|\rho^{k-T_s} +  c_wc_rb\|\Tilde{\theta}_0\|\sum_{i=T_s}^{k-1}\rho^{k-i-1}\lambda^{i-T_s},
\end{align*}
where the second inequality follows from the bound of $B_k\phi_k$ and Lemma \ref{lem:rls},  and we define $c_{T_s}:= \frac{1-\rho^{T_s}}{1-\rho}$. For $0<k\leq T_s$
   \begin{equation*}
    \|x_k-x_k^\star\| \leq c_w b\sum_{i=1}^{k}\rho^{k-i}\|\Tilde{\theta}_{i-1}\|\leq c_wb\|\Tilde{\theta}_0\|\sum_{i=1}^{k}\rho^{k-i}
    \end{equation*}
Similar to the proof of Theorem \ref{thm:rpl}, the states, $x_k$ and $x_k^\star$ can both be shown to be bounded and, hence
\begin{align*}
    &\mathcal{R}_T(\ak{u^{\mathrm{RLS}}})\leq \underbrace{L_c\sum_{k=T_s+1}^{T-1}\|x_k-x_k^\star\|}_{\text{I}} + \underbrace{L_c\sum_{k=0}^{T_s}\|x_k-x_k^\star\|}_{\text{ II}},
\end{align*}
where, as before,  II represents the cost in the initial $T_s$ long PE phase, and  I the cost improvement afterward. Using the Cauchy Product inequality and the inequalities above
\begin{align*}
    \text{ II}&\leq c_wbL_c\|\Tilde{\theta}_0\|\sum_{k=0}^{T_s-1}\sum_{i=0}^k\rho^{k-i}\leq c_wc_{T_s}T_sbL_c\|\Tilde{\theta}_0\|\\
    \text{I}
    &\leq c_wbL_c\|\tilde{\theta}_0\|\left(\rho c_{T_s}\!\sum_{k=1}^{T-T_s-1}\rho^{k}+\ak{c_r}\sum_{k=0}^{T-T_s-2}\sum_{i=0}^{k}\rho^{k-i}\lambda^{i}\right)\\
    &\leq \rho c_wbL_c\|\Tilde{\theta}_0\|\frac{\rho^T + \rho}{(1-\rho)^2}\\
    &\quad + c_wc_rbL_c\|\Tilde{\theta}_0\| \frac{\left(1-\rho^{T-T_s-1}\right)\left(1-\lambda^{T-T_s-1}\right)}{\left(1-\rho\right)\left(1-\lambda\right)}.
\end{align*}
Combining the terms completes the proof. The asymptotic stability proof parallels that of Theorem \ref{thm:rpl} and is  omitted.
\end{proof}

\ak{Unlike RPL, the RLSFF controller requires the stronger PE condition, but does not restrict the features as in Lemma 2. A sufficient condition for Theorem \ref{thm:rls} to hold is that of PE of $B_k\phi_k^\top$. Given that the input and feature matrices are time-varying, this condition can hold also in the limit, as the state converges to the origin. This happens, notably, in the case of tracking of time-varying references. A numerical example showcasing this is presented in the next section.}

\subsection{Analysis of the Bounds}

Both the RPL and the RLSFF adaptive controllers achieve a finite regret containing three separate terms of potential interest. A constant term, an exponentially decaying one with rate $\rho$, and a linear term in the PE time $T_s$. The latter is of most interest, which arises from the non-expansive properties of both estimators and signifies that no improvement in the worst-case can be expected until the time $T_s$ when the PE condition is met. 
Comparing the regret bounds directly,  when $\gamma <\lambda$ and $c_r>1$, i.e. $\varepsilon>  \delta$,  then the upper bound of the RLSFF regret is strictly larger than that of the RPL controller.

\section{Numerical Simulation}
\label{sec:numerical}
In this section, we show that the MRAC problem can be cast into the considered setting and provide a numerical example to showcase the performance of \ak{the} discussed adaptive controllers. Consider the  uncertainty-matched system  
\begin{equation*}
    x_{k+1} = Ax_k + B(u_k - \ak{\psi^\top}(x_k)\theta^\star),
\end{equation*}
where $A \in \mathbb{R}^{n\times n}$, $B \in \mathbb{R}^{n \times m}$,  \ak{$\psi: \mathbb{R}^n \rightarrow \mathbb{R}^{p\times m}$ is a feature matrix}, and $\theta^\star \in \mathbb{R}^p$ is unknown. The MRAC goal is to track the reference dynamics $
    \bar{x}_{k+1} = A_r \bar{x}_k + B_r r_k,$
where $\bar{x}_k \in \mathbb{R}^n$, is the ideal reference vector to be tracked, $r_k \in \mathbb{R}^m$,  is the bounded input command, $B_r\in \mathbb{R}^{n\times m}$ and  $A_r \in \mathbb{R}^{n \times n}$ is Schur stable. As in \cite{dogan2020improving}, we assume $(A,B)$ to be controllable and $B$ to be full column rank, and consider the input  $u_k = -K_1x_k +K_2r_k + \ak{\psi^\top(x_k)}\theta_k$ with matrices $K_1, K_2 \in \mathbb{R}^{m \times n}$  chosen such that $A-BK_1 = A_r$ and $BK_2 = B_r$ hold and $\theta_k$ being the parameter estimate. Denoting $e_k = x_k-\bar{x}_k$, \ak{the error dynamics are represented in the form of \eqref{eq:nonlinear_matched}}
\begin{equation}
    \label{eq:error_dynamics}
    e_{k+1} = A_r e_k + B\phi_k^\top (e_k)\tilde{\theta}_k,
\end{equation}
\ak{where we define $\phi_k^\top(e_k):= \psi^\top(e_k + \bar{x}_k)$}. Since $A_r$ is stable, the \acrshort{ediss} condition is satisfied. \ak{ We take the cost to be $c(e):= \|e\|^2$}, which is locally Lipschitz since the closed-loop states are bounded. \ak{ Note that, with this formulation, the SE/PE condition can hold even when the error dynamics are asymptotically stable, given that the reference ${r}_k$ is such that the feature matrices $\phi_k(e_k)$ are SE/PE. }

We consider the following example system \cite{dogan2020improving}
\begin{equation*}
    A = \begin{bmatrix}
        1.0314& 0.2526\\
        0.2526& 1.0314 
    \end{bmatrix}, \; B = \begin{bmatrix}
        0.0314\\
        0.2526
    \end{bmatrix},
\end{equation*}
with the reference model dynamics
\begin{equation*}
    A_r = \begin{bmatrix}
        -0.9929& 0.2253\\
        -0.0569& 0.8117 
    \end{bmatrix}, \; B_r = \begin{bmatrix}
        0.0314\\
        0.2526
    \end{bmatrix}.
\end{equation*}
The true parameter is $\theta^\star = \begin{bmatrix}
    0.75& 0.50
\end{bmatrix}^\top$ and \ak{$\psi(x_k) = x_k$, with the corresponding $\phi_k(e_k) = e_k + \bar{x}_k$}. Figure \ref{fig:mrac} shows the tracking of the first component of state $x_k$ for the above example under the RPL, RLSFF controllers and the command governor-based MRAC controller proposed in \cite{dogan2020improving} without regret guarantees. 
\begin{figure}
    \centering
\includegraphics[scale=0.17]{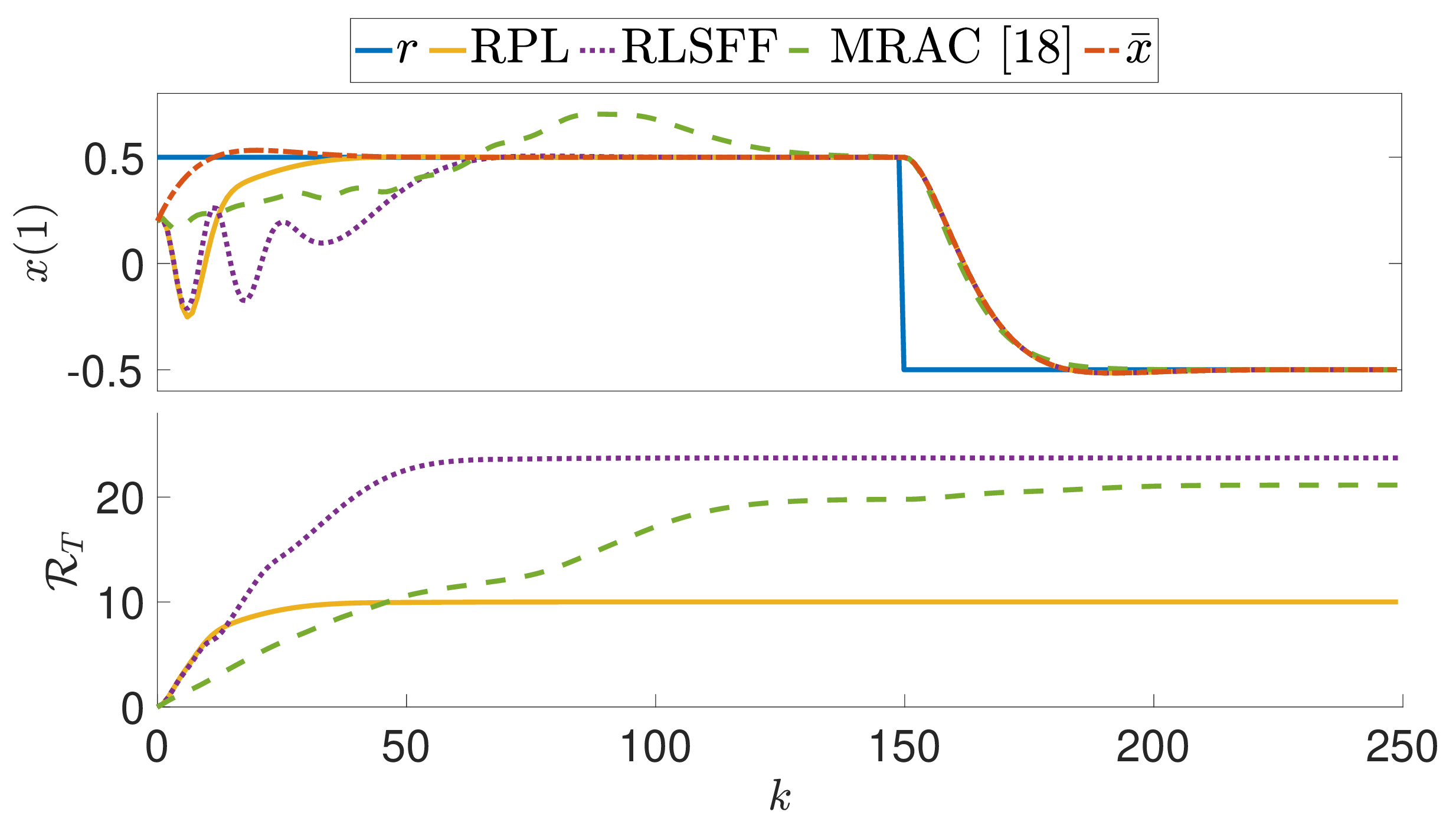}
    \caption{Comparison of adaptive controllers for the MRAC example in terms of asymptotic tracking (on the top) and regret (on the bottom).}
    \label{fig:mrac}
\end{figure}
For all we initialize the estimate by $\theta_0\!=\!\begin{bmatrix}
    5.00& -1.00
\end{bmatrix}^\top$, take $\varepsilon=1$ and start the simulation at $x_0 = \begin{bmatrix}
    0.2& 0.2
\end{bmatrix}^\top$. For RLSFF, we take $\lambda^2=0.99$, as lower values led to conditioning problems. For the controller of \cite{dogan2020improving} we used the same parameters as in \cite[Example 1]{dogan2020improving} and a tuned saturation value of $1.5$ for the command governor update for best results.
The superior performance of RPL is evident from \cref{fig:mrac}, despite a less stringent PE condition. The top figure shows the tracking performance of the controllers, \ak{and} the bottom, the regret for the quadratic costs, which in this case is $\mathcal{R}_T(u): = \sum_{k=0}^{T-1}\|e_k\|^2$, since by definition $\Bar{x}_0={x}_0$ and therefore $e^\star_k=0$ for all $k\in \mathbb{N}$.

\section{Conclusion}

Two recursive learning algorithms are considered for the nonlinear adaptive control problem with matched uncertainty in the context of regret minimization. We show that both the recursive least squares with forgetting factor and the novel recursive proximal learning algorithm are asymptotically stable and achieve finite regret scaling with the time required to achieve persistence of excitation. Possible extensions include the consideration of inexact basis matrices, bounded, non-stochastic noise, and goal-oriented controllers that instead of the estimation cost optimize over the objective cost directly.

\bibliographystyle{IEEEtran}
\bibliography{bibliography.bib}

\end{document}